\documentclass[conference]{IEEEtran}
\usepackage{graphicx}
\usepackage{epstopdf}
\usepackage{textcomp}
\usepackage{hyperref}
\usepackage{amsmath}
\usepackage{url}
\usepackage{color}
\newtheorem{definition}{Definition}
\newtheorem{proof}{Proof}
\newtheorem{theorem}{Theorem}

\usepackage{algorithm}
\usepackage{algorithmic}

\ifCLASSINFOpdf
\else
\fi

\renewcommand{\baselinestretch}{1.0}

\IEEEoverridecommandlockouts
\begin{document}
%
\title{Robust Component-based Network Localization with Noisy Range Measurements\thanks{This work was supported in part by the National Natural Science
	Foundation of China Grant No. 11671400, 61672524;
	the Fundamental Research Funds for the Central University,
	and the Research Funds of Renmin University of China,
	2015030273.
	Corresponding Author, Yongcai Wang, email: ycw@ruc.edu.cn}}

\author{\IEEEauthorblockN{Tianyuan Sun\IEEEauthorrefmark{1}, Yongcai Wang\IEEEauthorrefmark{1}, Deying Li\IEEEauthorrefmark{1}, Wenping Chen\IEEEauthorrefmark{1}, Zhaoquan Gu\IEEEauthorrefmark{2}}
\IEEEauthorblockA{\IEEEauthorrefmark{1}Department of Computer Sciences, Renmin University, Beijing, China, 100084}
\IEEEauthorblockA{Cyberspace Institute of Advanced Technology (CIAT), Guangzhou University, Guangzhou, China, 510006}}


%


\maketitle

\begin{abstract}
Accurate and robust localization is crucial for wireless ad-hoc and sensor networks. Among the localization techniques,  component-based methods advance themselves for conquering network sparseness and anchor sparseness. But component-based methods are sensitive to ranging noises, which may cause a huge accumulated error either in component realization or merging process.  This paper presents three results for robust component-based localization under ranging noises. (1) For a  rigid graph component,  a novel method is proposed to evaluate the graph's possible number of flip ambiguities under noises. In particular, graph's \emph{MInimal sepaRators that are neaRly cOllineaR  (MIRROR) } is presented as the cause of flip ambiguity,  and the number of MIRRORs indicates the possible number of flip ambiguities under noise. (2) Then the sensitivity of a graph's local deforming regarding ranging noises is investigated by perturbation analysis. A novel Ranging Sensitivity Matrix (RSM) is proposed to estimate the node location perturbations due to ranging noises.  (3) By evaluating component robustness via the flipping and the local deforming risks, a Robust Component Generation and Realization (RCGR) algorithm is developed, which generates components based on the robustness metrics. RCGR was evaluated by simulations, which showed much better noise resistance and locating accuracy improvements than state-of-the-art of component-based localization algorithms.  

\end{abstract}


%
\IEEEpeerreviewmaketitle

\vspace{0.1cm}

\section{Introduction}
A general scenario is There maybe few anchors to define the global coordinate system, and the distance measurements among the nodes maybe sparse. 
The problem of network localization has drawn great attentions\cite{liu_location_2010}. It is  to estimate the locations of nodes based on the inter-node range measurements, which generally adopts a graph realization model\cite{aspnes_theory_2006}.

Existing results on network localization can be roughly divided into two categories. The first category investigated the \emph{localizability problem} \cite{yang_beyond_2009}\cite{yang_understanding_2010}, i.e., given the settings of anchors and the distance matrix obtained from the network, researchers investigated whether the whole network or individual node can be uniquely localized or not. 

The second category focus on algorithms of network localization, which can be further divided into 1) \emph{semi-definite programing (SDP)} based, 2) \emph{trilateration-based}, and 3) \emph{component-based}. SDP-based method  adopts centralized optimization algorithm\cite{biswas_semidefinite_2006}\cite{so_theory_2006}. 
The trilateration-based and component-based method can be fully distributed\cite{moore_robust_2004}. In particular, the trilateration-based  method calculates the locations of some localizable nodes by anchors, and sequentially localizes the other nodes by treating the newly located nodes as anchors\cite{goldenberg_localization_2006}\cite{moore_robust_2004}. Its main limitation is that each localizable node need to be directly connected to three anchors, which is actually not necessary in the optimal case\cite{yang_understanding_2010}\cite{aspnes_theory_2006}. Error accumulation is also a main drawback of trilateration-based  method because of the sequential localization process\cite{yang_quality_2010}. 

Component-based methods, recently proposed in \cite{cucuringu_sensor_2012}\cite{wang_etoc:_2010}\cite{wang_component_2008}\cite{wang_component-based_2011}\cite{_component-based_2015} partition the nodes into rigid components. 
Locations of nodes in each rigid component are calculated by trilateration or biliteration. Multiple realization candidates were preserved when the component is not uniquely localizable.  Then each component  is treated as a basic unit. 
Edges between components  help to merge the units and anchors in different units can collaborate to convert the local coordinate systems of components to a global system. As a result, component-based method can localize the nodes which are not localizable in trilateration, and it reduces error accumulation by the collaboration of distributed components. 

But both the localizability problem and the localization algorithms face challenges when distance measurements are noisy, which is inevitable in real applications.  The essential difficulty is that graph rigidity analysis is based on bars and joints with exact lengths\cite{jackson_notes_2002}, which is the key for guaranteeing localization solution uniqueness. 
But, because of the noises of distance measurements, the true location of a node maybe at a point near the estimated location, or at a flipped position (called by flip ambiguity) far from the estimated location. The solution uniqueness becomes difficult to verify under ranging noises. Furthermore, the noises can seriously worsen the error accumulation in sequential localization process.   

To deal with noises, bounding the worst-case locating error becomes an important problem in existing studies\cite{wang_etoc:_2010}\cite{moore_robust_2004}\cite{xiao2015noise}\cite{xiao2018noise}. 
In trilateration-based method, \cite{moore_robust_2004} proposed robust four-node quadrilateral 
as the smallest possible unit whose flip ambiguity probabilities can be bounded. 
In \cite{yang_quality_2010}, trilateration confidence was proposed to evaluate  the quality of trilateration.  But error accumulation is still serious  in sequential localization, even if we select robust quadrilateral or confidential trilateration in each step. 
\cite{wang_etoc:_2010} investigated how to bound the worst-case location error in component-based localization by robust component merging. They proposed efficient algorithm to achieve error-bounded component mergence for  four special scenarios.  Existing practical indoor localization approaches  also proposed dynamic time warping method \cite{ye_accurate_2018} and radio signal strength based voronoi method \cite{ye_robust_2017} for conducting robust localization to dealing with ranging and radio signal strength noises. 


But, due to error accumulation in sequential localization, how to guarantee the localization error in a component to be small is still challenging. It is highly dependent on the component topology. The local error in a component may include both discontinuous flipping errors and continuous deforming errors. To tackle error accumulation in component merging, methods are needed to evaluate the reliability of the realized component. 
Efficient metrics to evaluate the trustworthy of components and efficient methods to avoid the error accumulation are highly desired. 

This paper presents novel quantitative metrics to evaluate the risk of flip ambiguity and the sensitivity of local deforming of a realized component under noises, and efficient algorithm to avoid error accumulation by utilizing the proposed quantitative metrics.  The main contributions are in three aspects:

1) For a generated rigid component under noises, efficient method was proposed to  
find  possible number of flip ambiguities.  \emph{MInimal sepaRator that are neaRly cOllineaR  (MIRROR) } is proposed as the potential cause of  flip ambiguity. Each MIRROR is modeled by a band for noise tolerance, which indicates the nearly collinear vertex separator of the graph. Efficient algorithms are proposed to detect the MIRRORs. The number of MIRRORs indicates the number of possible flip ambiguities. 

2) The sensitivity of location deforming regarding to ranging noises is investigated by perturbation analysis. A novel Ranging Sensitivity Matrix (RSM) is proposed,  which builds a linear equation between the location perturbation of nodes and the ranging noises.  The condition number of RSM indicates the sensitivity of the graph's local deforming regarding to the  ranging noises.  

3) By integration of the flip ambiguity and the local deforming sensitivity metrics, a Robust Component Generation and Realization (RCGR) algorithm is developed. It realizes components selectively based on trustworthy to reduce error accumulation. 
RCGR  is implemented and verified by simulations. It forms  more robust components than state-of-the-art method, leading to better location accuracy, which is nearly comparable to the centralized  approach. 
 \begin{figure*}[htbp]
 \centering
 \begin{minipage}[t]{0.32\linewidth} 
\begin{center}
 \includegraphics[height=1.6in]{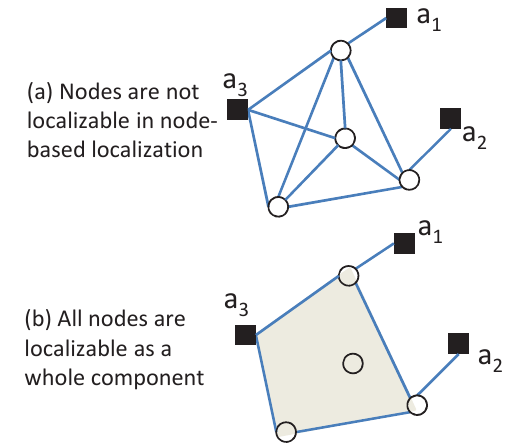}
 \caption{Component-based localization vs. node-based localization}
 \label{nodecomp}
 \end{center}
 \end{minipage}
 \begin{minipage}[t]{0.32\linewidth} 
\begin{center}
 \includegraphics[height=1.5in]{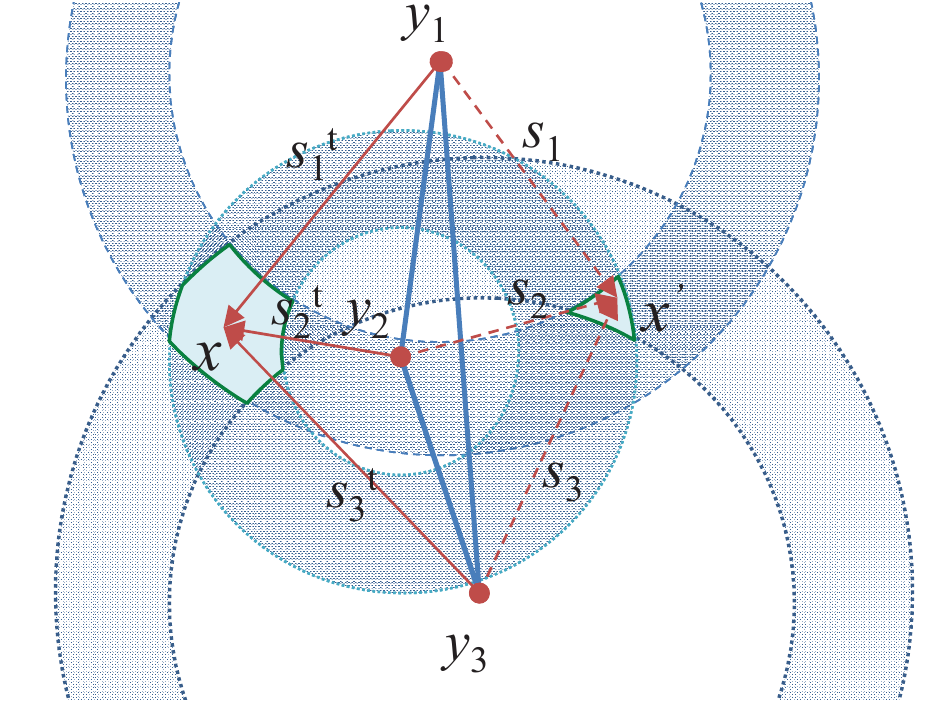}
 \caption{Discontinous ambiguous regions caused by noises}
 \label{discontinuous}
 \end{center}
 \end{minipage}
  \begin{minipage}[t]{0.32\linewidth} 
 \begin{center}
  \includegraphics[height=1.5in]{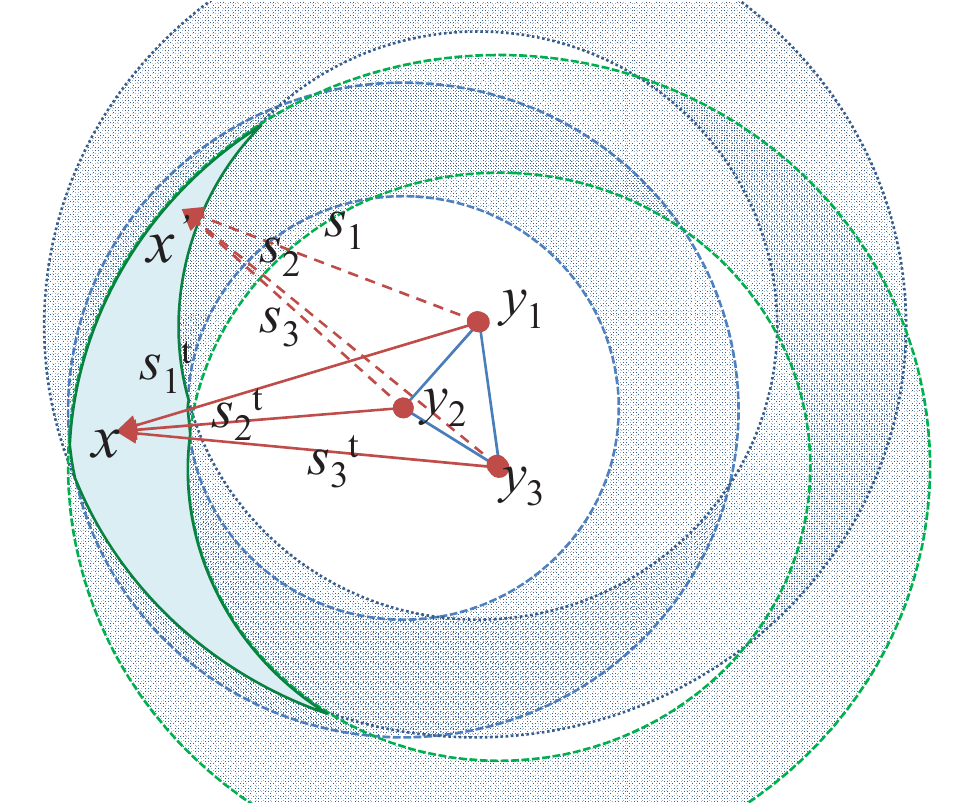}
  \caption{Continuous ambiguous region caused by noises}
  \label{continuous}
  \end{center}
  \end{minipage}
 \vspace{-0.5cm}
 \end{figure*}

The rest sections of this paper are organized as following. The preliminaries are introduced in Section II. Flip ambiguity analysis is presented in Section III. Sensitivity of local deforming is analyzed in Section IV. RCGR  is introduced in Section V. Simulation results are presented in Section VI and conclusion with remarks are presented in Section VII. 
 \section{Preliminary}
\subsection{Terminology}
The network to be localized is described by a  graph $G=(V, E)$, where each vertex $v_i \in V$ denotes a node and each undirected edge $(i,j)\in E$ stands for a distance measurement between two nodes $v_i$ and $v_j$. The measured distance is denoted by  $d_{i,j}$. Distance between $i$ and $j$ can be measured if $d_{i,j} <R$, i.e., the maximum ranging distance.  
No specific distribution is assumed to the ranging noises,  but we can assume the measurement noises have high probability to be smaller than an upper bound $C$. 

A small portion of nodes have calibrated locations, which are anchors, denoted as $\{v_1, v_2, \cdots, v_m\}$. Distances among all anchors are perfectly known without noise. The other $n-m$ nodes are ordinary nodes  to be localized. The estimated location of nodes are denoted by $\{\mathbf{p}_i, i=m+1, \cdots, n\}$. The problem is to estimate the locations of the ordinary nodes, so that the distances among all the vertexes best match the noisy distance measurements. 
\subsection{Network Localization Approaches} 
We briefly review the component-based localization methods, and the approaches for localization robustness.   
\subsubsection{Component-based Localization}
Component-based localization contains three steps: \emph{component generation, component realization, and component merging}\cite{wang_component_2008}\cite{wang_component-based_2011}. 
\begin{enumerate}
\item The component generation is to partition the network into a set of rigid components. Each component is initialized by a triangle (not necessarily to be anchors), which initializes a local coordinate system. Then other nodes are added if there are two edges connecting the node to the component.  
\item Component realization localizes nodes in the component sequentially by trilateration or bilateration. It keeps multiple realization candidates of the component if the realization solution is not unique.  
\item Component merging is to merge the components to one coordinate system, by the anchors in each component and the links between the components. It treats each component as a rigid body merge them by transition and rotation\cite{wang_component-based_2011}. So nodes in all components can be localized. 
As anchors integrate in the mergence process,  which may make two unrealizable components be a realizable one.   
\end{enumerate}

The advantage of  component-based approach is that it has capability to localize unlocalizable nodes in trilateration-based localization. As an example shown in Fig.\ref{nodecomp}. The black nodes are anchors. No nodes can be localized by trilateration because any node has at most two neighboring anchor. But the component can be localized as a whole, since there are three anchors to locate this rigid component. 

  
\subsection{Approaches To Deal With Noises}
Measurement noises are inevitable in real applications. Tackling noises is critical problem in network localization. 
\subsubsection{In Node-based Approach}
Moore et al. \cite{moore_robust_2004} proposed robust four-vertex quadrilateral as the smallest possible subgraph that can avoid flip ambiguity error under noise, so robust quadrilateral is selected for trilateration in each step. Yang et al. \cite{yang_quality_2010} proposed quality of trilateration. The trilaterations with the highest quality should be selected in each step. 
But note that the error accumulations  may still cause large error in the later steps. 
\subsubsection{Dealing Noises in Component-based Localization} \cite{wang_etoc:_2010} studied robust component merging problem under noises. They characterized four patterns in which components can be robustly merged. But an assumption of their approach is that the node locations in each component are trustworthy under noises. This is hard to satisfy in practice because of the error accumulation in sequentially localizing the nodes in the components.  

\subsection{Type of Errors}
In a rigid graph, ranging noises may cause either \emph{discontinuous or continuous localization errors}. 
Since the upper bound of noise is denoted by $C$,  the location of $v_i$ should satisfy the following constraints:
\begin{equation}
\begin{array}{l}
\forall j{\rm{,if  }}{\rm{  }}~~\exists {\rm{ }}(i,j) \in E\\
{\rm{s.t }} - C \le {\left\| {{x_i} - {x_j}} \right\|_2-d_{i,j}} \le C
\end{array}
\label{eqn1}
\end{equation}

That is, $x_i$ is in the intersected region of the  annulus regions centered at $\{x_j, \forall (i,j)\in E\}$, with radius in the range $[d_{i,j}-C, d_{i,j}+C]$. This region is called \emph{ambiguous region  (AR) }  of $x_i$. Fig.\ref{discontinuous} and Fig.\ref{continuous} give examples  to illustrate how the estimation results may be blurred by noise. 

In Fig.\ref{discontinuous}, when the reference points are nearly on a common line, the ambiguous region is separated into two isolated sub-regions. In this case, if $x'$ is the estimated location by the noisy measurement and $x$ is the true location, there is \emph{discontinuous location error} between $x$ and $x'$, which is called flip ambiguity\cite{jackson_notes_2002}. 

In  Fig.\ref{continuous}, when the reference points are very close comparing to the ranging distance, the ambiguous region is long and narrow. The true state $x$ may be deformed to $x'$, which is called \emph{``continuous local deforming error"}.  
For the worst-case localization error of $x$ to be bounded, we need to avoid the flip ambiguity and to avoid large local deforming error. 
\subsection{Error Accumulation}
When the component is generated sequentially, error accumulation is a serious problem, which may cause  later located node  suffer huge localization error. 

Fig.\ref{accum} shows an example of the impact of error accumulation.  Fig.\ref{accum}a) shows true positions of five nodes and their inter-node distance measurements. 
Suppose we have known the locations of $a$, $b$, $c$ as reference points. 
Fig.\ref{accum}b) shows the result when $d$ is trilaterated by $a$, $b$, $c$ in the first step and $e$ is trilaterated by $a$, $b$, $c$, $d$ in the second step. Since $d$ has little triangulation error, the result of $e$ also has small location error. Fig.\ref{accum}c) shows the result when $e$ is calculated firstly and $d$ is calculated secondly. Since the location error of $e$ is large, the location result of $d$ is also large caused by the location error of $e$.

Therefore, to reduce the accumulated error in the component generation, we should select the nodes with small risk of error with high priority.  At the same time, in component mergence process, the components with low risk of flip ambiguity and small local deforming error should be chosen with higher priority. 
In the following two sections, novel methodologies to evaluate the risk of flip ambiguity and the sensitivity of local deforming for a given rigid graph component are presented, which are foundation for robust network localization methods.

\section{Risk of Flip Ambiguity}
Considering a generated rigid graph component under noise, we consider the problem to evaluate its risk of flip ambiguity. 
The locations of some ordinary nodes may flip together, i.e., a subgraph may flip. The risk of flip should be evaluated on the graph  level. 
Therefore, we propose an novel and efficient method, which is to find all the \emph{MInimal sepaRators that are neaRly cOllineaR  (MIRROR) in the graph}. A subgraph may flip over a MIRROR to a wrong position. The total number of MIRRORs  that can be found in the graph indicates the possible number of flip ambiguities of the considered graph.  

A minimal separator $S$ of a graph $G$, is a subset of vertexes whose removal separates $G$ into at least two disconnected components. 
A separator is said to be a minimal separator if the removal of any vertex in $S$ will cause $S$ no longer a separator\cite{kumar_minimal_1998}. Listing all minimal separators is one of the fundamental enumeration problem in graph theory, which has great practical importance in network reliability analysis \cite{kumar_minimal_1998}\cite{provan_complexity_1983}. 

\begin{figure}[t]
\begin{center}
\includegraphics[width=0.4\textwidth]{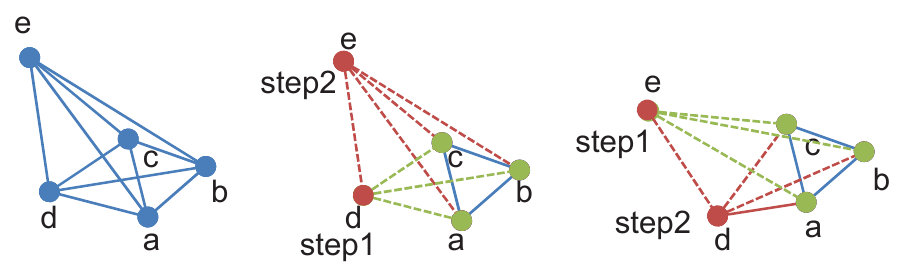}
\caption{ (a) true locations of five nodes and their inter-node distances.  (b) $d$ is generated in the first step and $e$ is in the second step; (c) $e$ is generated in the first step and $d$ is in the second step}
\label{accum}
\end{center}
\end{figure}

\subsection{MIRRORs Under Noise}
First, we look at an example to understand the  problem. In a rigid graph component shown in Fig.\ref{mirror}(a), we check where flip ambiguity can happen. Since noises are presenting, the generated locations may vary in a small neighborhood when noise values change. Therefore, as shown in  Fig.\ref{mirror}(b), we can find three potential MIRRORs, i.e., $\{1,9,5\}$, $\{4,1,8\}$, $\{2,1,6\}$, across which the graph can flip.  Note that $\{8, 6\}$ is also a minimum vertex separator,  but because ${7}$ is nearly collinear with $\{8, 6\}$ and they are on the boarder of the graph, the flip of the whole graph over $\{6,7,8\}$ is a rigid transformation. So $\{8, 6\}$ should not be counted as a MIRROR. 




Therefore,  an efficient algorithm is proposed to find MIRRORs  under noise. 
It can be seen that, there are two conditions to detect a MIRROR: 1) the vertexes on a MIRROR are nearly collinear under noise;  and 2) the removal of them can separate the graph into two disconnected components. Traditional methods that find the minimal vertex separators by graph theory\cite{chandran_linear_2006} did not consider the collinear property nor the noise impacts. In this paper, each noisy, nearly collinear vertex separator is modeled by a \emph{band}, whose width indicates the tolerance to the noises. Efficient algorithm is proposed to find the bands, whose removal can separate the graph, which are the detected MIRRORs. 

\subsection{MIRROR  Finding Algorithm}
\subsubsection{Band Generation and Mergence} 
Note that every two points can determine a line segment. The $n$ nodes with known coordinates can form $n(n-1)/2$ line segments. Since the nodes' locations are noisy, the width of the line segments is increased  to convert the line to a \emph{band}. Each band $B_i$ is centered on the line, determined by the coordinates of the firstly added two nodes,  and has width $2c$.  $c$ indicates the location error range and the center line is denoted by $l(B_i)$.  Then vertexes covered by a band are nearly collinear and will be merged into one band.  After band generation and mergence, each band is charactered by a set of nearly collinear vertexes, i.e., $B_i={v_{i,1}, ...v_{i,|B_i|}}$. 

\subsubsection{Verifying whether a band is a MIRROR} 
For any band, we evaluate whether it is a separator. We treat the band as a cut by removing the vertexes in the band and the edges connected to the removed vertexes. The other nodes that are not in the band may either lie on one side of the band,  or lie on two sides of the band.  The band is a MIRROR if the following three conditions are all satisfied: 
\begin{definition}[MIRROR Detection Conditons]
\begin{enumerate}
\item The other nodes lie on two sides of the band. 
\item There is no edge connecting the two separated groups.
\item At least one separated group does not have an anchor.  
\end{enumerate} 
\end{definition}

\begin{figure}[t]
\begin{center}
\includegraphics[width=0.4\textwidth]{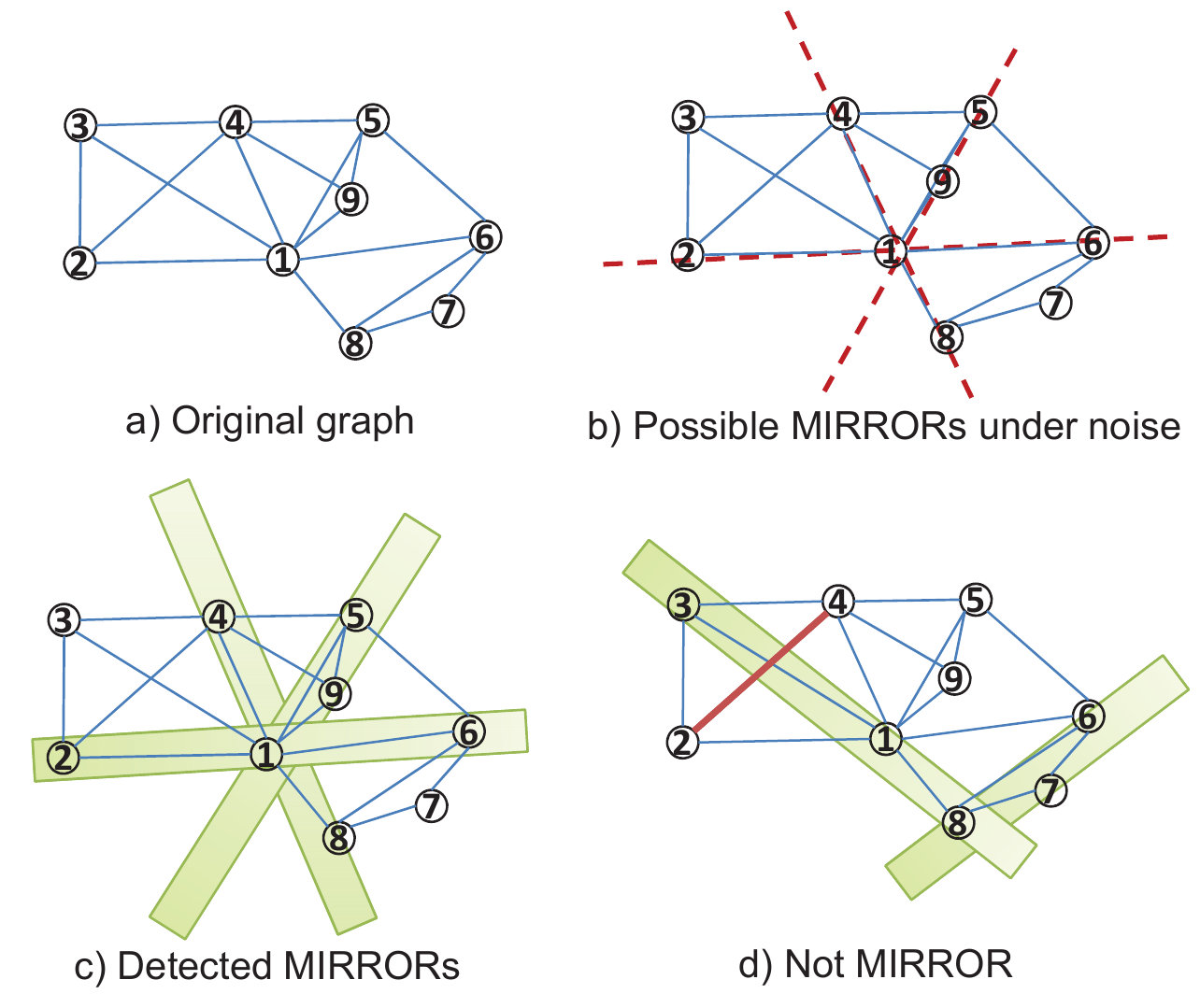}
\caption{ (a) a rigid component generated by noisy measurements.  (b) mirrors formed by nearly collinear vertex cuts. }
\label{mirror}
\end{center}
\end{figure}

\subsubsection{Example} 
An example is given to illustrate these conditions.  The three MIRRORs found in  graph of Fig.\ref{mirror}(a) are illustrated by the three bands in  Fig.\ref{mirror}(c) . It can be seen that by removing the vertexes in each band, the graph is separated into two disconnected groups. There are no edges connecting the two groups and no anchor exists in each group.  Fig.\ref{mirror}(d) shows two examples when the band is not a MIRROR.  For the band characterized by $\{3,1,8\}$, its removal cannot separate the graph,  since $\{2\}$ is connected to the graph by $\{4\}$.   The band characterized by $\{6,7,8\}$ is not a MIRROR, because all the other nodes lie on one side of it. 

\subsection{Algorithm Details and Analysis}
\subsubsection{Band Generation and Mergence Algorithm}
Let $B=\{B_1, \cdots, B_k\}$ denote the set of $k$ detected bands. 
Let $\mathcal{N}(B_i)$ return all the nodes in the band $B_i$, i.e., nodes covered by $B_i$. Let $B_c(i)$ return the set of bands  containing node $i$.  Let  $d(v_i, B_j)$ indicate the distance from a node $v_i$ to a band $B_j$, i.e., the minimum distance from the point $v_i$ to the center line $l(B_j)$. Then for a node $v_j$, the minimum distance from $v_j$ to any band that covers $v_i$ can be expressed as:
\begin{equation}
[d_{j,b_i}, b_i] = \textrm{find\_min\_dis}(v_j, B_c(i))
\end{equation}
where $b_i$ is the returned band covering $v_i$ and it is closest to $v_j$. $d_{j,b_i}$ is the returned minimum distance from $v_j$ to $b_i$. 
 Then the algorithm of band generation and mergence is shown in Algorithm~\ref{algorithm1} . 
\begin{algorithm}
\caption{Band Generation and Mergence}
\begin{algorithmic}[1]
\REQUIRE $G=(V,E)$ and $\{\mathbf{p}_i, i=1,\cdots, n\}$ the node locations 
\ENSURE $B$ the set of detected bands. 
\STATE  Initialize $B=\emptyset$; $k=0$ 
\FOR{$i =1:1:n$}
\FOR{$j = i+1:1:n$}
    \STATE  $B_c(i)$ = find bands covering $i$. 
    \STATE  $[d_{j,b_i}, b_i] $ = $\textrm{find\_min\_dis} (v_j, B_c(i))$ 
	\IF{$B==\emptyset$ \OR $d_{j,b_i} > c$}
	   \STATE   $k=k+1$; $B_k=\{v_i,v_j\}$, $B= B\cup B_k$ 
	\ELSE 
	   \STATE   $B_{b_i} = B_{b_i} \cup v_j$ 
	\ENDIF
\ENDFOR
\ENDFOR	  
\end{algorithmic}
\label{algorithm1}
\end{algorithm}

In line 7, if the distance from $v_j$ to all bands covering $v_i$ is larger than $c$,  $\{v_i,v_j\}$ is generated as a new band. Otherwise $v_j$ is added into the closest band covering $v_i$, i.e., $v_j$ is merged into $B_{b_i}$. We can see the worst case complexity of the algorithm is $O(n^4)$, for the worst case complexity of finding bands covering $j$ is $O(n^2)$ where $n$ is the number of nodes. 

\subsubsection{MIRROR  Verification Algorithm}
Given the generated bands $B$, the algorithm verifies whether each band $B_i$ is a MIRROR  by removing  $\mathcal{N}(B_i)$ and the edges connected to   $\mathcal{N}(B_i)$. It checks whether the remaining graph satisfies the MIRROR detection conditions. 

Since the center line  $l(B_i)$ of band $B_i$ is characterized by the firstly added two nodes. Let's denote the coordinates of these two nodes $\mathbf{u}_{i}$ and $\mathbf{w}_{i}$ respectively.  Considering another node with coordinates $\mathbf{p}_{j}$. The following sign function indicates the relative position between  $\mathbf{p}_{j}$ and  $l(B_i)$:
\begin{equation}
p\!\!=\!\!(\mathbf{w}_{i,x}\!-\!\mathbf{u}_{i,x})(\mathbf{p}_{j,y}\!-\!\mathbf{u}_{i,y})\!-\!(\mathbf{w}_{i,y}\!-\!\mathbf{u}_{i,y})(\mathbf{p}_{j,x}-\mathbf{u}_{i,x})
\label{sign}
\end{equation}
If $p=0$, $\mathbf{p}_j$ is on the line. $\textrm{sign}(p)=+1$ means $\mathbf{p}_{j}$ is on one side, and $\textrm{sign}(p)=-1$, if $\mathbf{p}_j$ is on the other side. By this method, condition 1) can be efficiently checked; and then condition 2) and condition 3) can be  checked easily. The algorithm details are given in Algorithm\ref{algorithm2}.   
 \begin{algorithm}
 \caption{MIRROR  Detection}
 \begin{algorithmic}[1]
 \REQUIRE  $G=(V,E)$, $\{\mathbf{p}_i, i=1,\cdots, n\}$, and $B$ 
 \ENSURE $M$: indicator of whether each $B_i$ is a MIRROR
 \STATE  Initialize $M = I_k$ as a length-$k$ one vector. 
 \FOR{$i =1:1:k$}
 \FOR{$j \in V\setminus B_i$}
      \STATE  calculate the sign of $p_j$ by (\ref{sign}) .
  \ENDFOR 
      \STATE  Let $P+$, $P-$ be the point set of $+$ ($-$) signs; 
 	 \IF{$|P+|==0$  \OR $|P-|==0$}
 	   \STATE   $M_i = 0;$ 
 	\ELSIF {$\forall i \in P+ $  and $j \in P-$, $\exists e_{i,j} \ne 0$}
 	   \STATE   $M_i = 0;$ 
    \ELSIF { both $P+ $  and $P-$ contain anchors}
 	   \STATE   $M_i = 0;$ 
 	  \ENDIF
 \ENDFOR
 \STATE  return $M$ and the zero norm of $M$, i.e., $||M||_0$;
 \end{algorithmic}
 \label{algorithm2}
 \end{algorithm}

 The algorithm initializes $M$ as an all one vector with the same length to $B$, and assigns $M_i$ to zero if any MIRROR detection condition is not satisfied for $B_i$. The sum of remaining ones in $M$, indicated by $||M||_0$ gives the number of detected MIRRORs.  Since the most possible number of bands is $n(n-1)/2$, the worst case complexity of Algorithm\ref{algorithm2} is $O(n^4)$, because for each band we need to check all links between the separated point sets.  
So the number of possible flip ambiguities of a given graph is indicated by $||M||_0$, i.e., the detected number of possible MIRRORs.


\section{Severity of Local Deforming }
In addition to the discontinuous deforming, continuous local deforming may also cause large error when noises are presenting. As shown in  Fig.\ref{continuous}, due to topology illness of the graph, the realized location $x$ may derivate to $x'$ due to small ranging noise.  But evaluating the local deforming severity of a graph by geometric-based method is not efficient.  It involves to solve a large set of quadratic inequality functions as stated in (1) to characterize the feasible regions of each node. Note that the feasible region of a node can be concave, which makes directly characterizing the feasible regions difficult. 
  
\subsection{Overview and Notations}
In this section,   an efficient matrix perturbation-based method is proposed to evaluate the local  deforming sensitivity of nodes regarding to  ranging noises. It is derived to use only the locations of the nodes and the edge lengths to construct a \emph{ranging sensitivity matrix (RSM)}, and the \emph{condition number} \cite{cline_estimate_1979} of the RSM can well predict the sensitivity of local deforming over the ranging noise. Previous results investigated stiffness matrix to evaluate the structure's resistance to noise\cite{zhu_stiffness_2009}, which considers the impacts of node movements. But it is the first time that RSM 　is proposed for evaluating graph robustness regarding to the ranging noises.

Since the edge lengths are noisy, the generated locations of nodes may deviate from the true locations. What we are interested is how serious may the deviation be due to the ranging noise. We present a perturbation-based analysis to evaluate how the variations of the edge lengths within a noise range may change the possible locations of the ordinary nodes. 



Let's consider the generated graph $G=(V,E)$. Suppose there are $m$ anchors in the component and $n-m$ ordinary nodes. Recall that the generated locations of these nodes are denoted by $\{\mathbf{p}_i \in \Re^2, i=1, \cdots, n\}$. The observed edge lengths are denoted by $\{d_{i,j} \in \Re,  (i,j)\in E\}$. Let vector ${{\mathbf{\vec e}_{i,j}}} = \mathbf{p}_i-\mathbf{p}_j$; Then ${d_{i,j}} = ||{\mathbf{\vec e}_{i,j}}|{|_2}$. The perturbation of edge lengths caused by noises are denoted by $\{\Delta d_{i,j} , (i,j)\in E \}$, and the perturbations of nodes' locations caused by edge length varations are denoted by  $\{\Delta \mathbf{p}_{i}, i\in \{1, \cdots, n\}\} $. The new locations of nodes are denoted by $\mathbf{p}'_i=\mathbf{p}_i+\Delta \mathbf{p}_{i}$. The new edge vector is denoted by ${{\mathbf{\vec e}_{i,j}'}} = \mathbf{p}'_i-\mathbf{p}'_j$. 
Note that the locations of the anchors are perfectly known and  $\{\Delta \mathbf{p}_{i}$ is always zero for $ i=\{1, \cdots, m\} $. So we evaluate whether the potential edge length variations will severely change the locations of the ordinary nodes. 

\subsection{Location's Sensitivity to Edge Length Perturbations}

Let's firstly consider one edge $d_{i,j}$ and its two end vertexes $\mathbf{p}_i$ and $\mathbf{p}_j$. If we change the length of $d_{i,j}$ by a sufficiently small amount $\Delta d_{i,j}$, the new locations of $i$, $j$ change to $\mathbf{p}_i'$, $\mathbf{p}_j'$. Let's investigate the relationship among $\Delta d_{i,j}$ and  $\Delta \mathbf{p}_{i}$ and  $\Delta \mathbf{p}_{j}$. 


\begin{figure}[t]
\begin{center}
\includegraphics[height=1in]{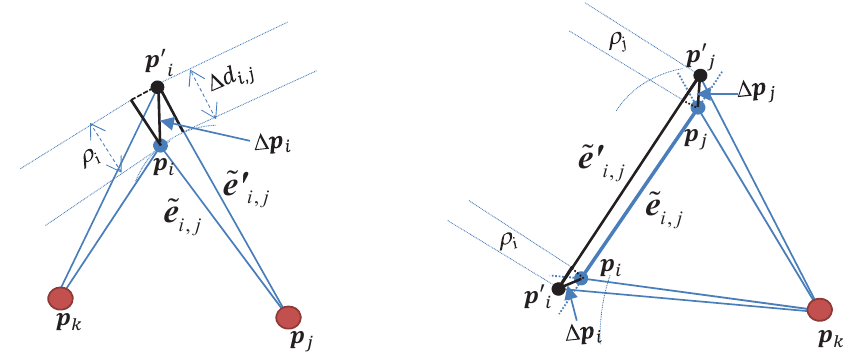}
\caption{ Example of node location perturbation subject to edge length changes. }
\label{perturbation}
\end{center}
\end{figure}

As shown in Fig.\ref{perturbation}(a), suppose $k$, $j $ are anchors, and $i$ is an ordinary node. $\mathbf{p}_i$ is the original location before edge length perturbation and $\mathbf{p}_i'$ is the deviated location of $i$ driven by the edge length variations of $d_{i,j}$ and $d_{i,k}$.  $||\Delta \mathbf{p}_i||_2$  is the derivated distance of $i$. Let ${\rho _{i}}\in \Re$ be the projected length of $||\Delta \mathbf{p}_i||_2$ onto the direction of  $\mathbf{\vec e}_{i,j}$. By linear projection, it can be expressed as:
\begin{equation}
{\rho _{i}} = \frac{{\left\langle {{\mathbf{\vec e}_{i,j}},\Delta \mathbf{p}_i} \right\rangle }}{{{{\left\| {{\mathbf{\vec e}_{i,j}}} \right\|}_2}\left\| {\Delta {\mathbf{p}_i}} \right\|_2}}\left\| {\Delta {\mathbf{p}_i}} \right\|_2 = \frac{{\left\langle {{\mathbf{\vec e}_{i,j}},\Delta \mathbf{p}_i} \right\rangle }}{{{d_{i,j}}}}
\label{pij}
\end{equation}
where ${\left\langle {{\mathbf{\vec e}_{i,j}},\Delta\mathbf{p}_i} \right\rangle }$ is the inner product of  $\mathbf{\vec e}_{i,j}$ and $\Delta \mathbf{p}_i$. 

Fig.\ref{perturbation}(a) illustrates the geometric relationship between $\rho_{i}$ and $\Delta d_{i,j}$. Since the perturbation is small, the angle $\angle {{\mathbf{p}'}_i}{\mathbf{p}_j}{\mathbf{p}_i}$ is small. So: 
 \begin{equation}
 \Delta d_{i,j}  \approx  \rho _{i} =\frac{{\left\langle {{\mathbf{\vec e}_{i,j}},\Delta \mathbf{p}_i} \right\rangle }}{{d_{i,j}}}
 \label{diff1}
 \end{equation}
 This sets up a linear relationship between $\Delta d_{i,j}$ and $\Delta \mathbf{p}_i$. Further, if both $i$, $j$ are ordinary nodes, the geometric relationship is shown in Fig.\ref{perturbation}(b). Since   $\|\Delta \mathbf{p}_i\|_2$ and $\|\Delta \mathbf{p}_j\|_2$ are small, $\mathbf{\vec e}_{i,j}$ and $\mathbf{\vec e}_{i,j}'$ are nearly parallel. So: 
 \begin{equation}
 \Delta d_{i,j} \approx \rho _{i}+\rho _{j} =\frac{{\left\langle {{\mathbf{\vec e}_{i,j}},\Delta \mathbf{p}_i} \right\rangle }}{{d_{i,j}}} +\frac{{\left\langle {{\mathbf{\vec e}_{j,i}},\Delta \mathbf{p}_j} \right\rangle }}{{d_{i,j}}}
 \label{deltad}
 \end{equation}
 Note that $\mathbf{\vec e}_{i,j}=\mathbf{p}_i-\mathbf{p}_j = -\mathbf{\vec e}_{j,i}$, so (\ref{deltad}) can be rewritten as 
 \begin{equation}
\Delta {d_{i,j}} = \left\langle {\frac{{{{\mathbf{\vec e}}_{i,j}}}}{{{d_{i,j}}}},\Delta { \mathbf{p}_i}} \right\rangle  - \left\langle {\frac{{{\mathbf{\vec e}_{i,j}}}}{{{d_{i,j}}}},\Delta { \mathbf{p}_j}} \right\rangle 
\label{diff}
 \end{equation}

The fact behind(\ref{diff1}) (\ref{diff}) is that, under small perturbation of $\Delta d_{i,j}$, $\mathbf{\vec e}_{i,j}$ and $\mathbf{\vec e}_{i,j}'$ are always nearly parallel, which also holds in more general cases. So the projection of $||\Delta { \mathbf{p}_i}||_2$ onto ${\mathbf{\vec e}'_{i,j}}$ can be approximated by the projection onto ${\mathbf{\vec e}_{i,j}}$, which can be calculated by the original coordinates of point $i$ and $j$. 
Therefore, for a graph containing $|E|$ perturbable edges, we can get $|E|$ equations similar to (\ref{diff}). Note that they don't include edges among anchors, since their lengths are constant.

Let $\Delta p_{i,x}$ and  $\Delta p_{i,y}$ be the deviations of  $\Delta \mathbf{p}_{i}$  in the  $x-$ and $y-$ directions respectively, and let $p_{i,x}, p_{i,y}$ be the $x-$ and $y-$ coordinates of $\mathbf{p}_i$. The inner product can then be expended and the $|E|$ equations can be written as the equation shown in Fig.\ref{equation}.
\vspace{-0.5cm}
\begin{figure}[htbp]
\begin{center}
\includegraphics[height=1.6in]{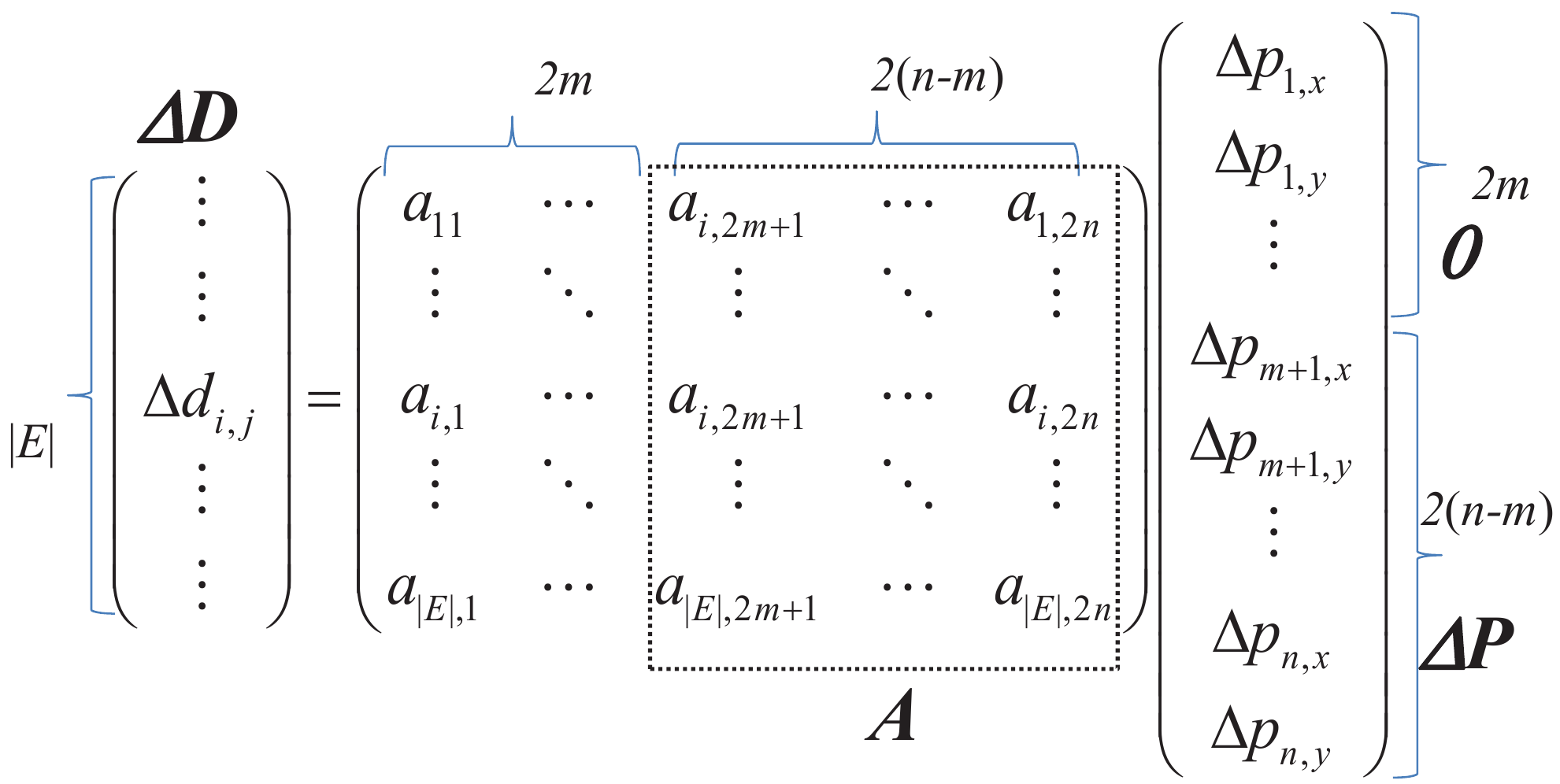}
\caption{Equation  generated by perturbation analysis of $|E|$ edges }
\label{equation}
\end{center}
\end{figure}
\vspace{-0.5cm}
\subsection{Ranging Sensitivity Matrix} 
The coefficient matrix has $|E|$ rows and $2n$ columns. Each row is corresponding to an edge. Every two columns $2i-1$, $2i$ are corresponding to the $x$ and $y$ coordinates of a vertex $i$. Since each edge has only two vertexes, there are at most four non-zero entries in each row. Suppose the $k$th row is corresponding to $\Delta d_{i,j}$, if neither $i$ nor $j$ is an anchor, there are: 
 ${a_{k,2i - 1}}=\frac{{{p_{i,x}} - {p_{j,x}}}}{{{d_{i,j}}}}$; $ {a_{k,2i}}=\frac{{{p_{i,y}} - {p_{j,y}}}}{{{d_{i,j}}}}$; ${a_{k,2j - 1}}= - \frac{{{p_{i,x}} - {p_{j,x}}}}{{{d_{i,j}}}}$; ${a_{k,2j}} =  - \frac{{{p_{i,y}} - {p_{j,y}}}}{{{d_{i,j}}}}$; and ${a_{k,v}} = 0,{\rm{ if ~~ }} v \ne \{ 2i - 1,2i,2j - 1,2j\} $. If $i$ (or $j$) is an anchor,  ${a_{k,2i}}=0$, ${a_{k,2i-1}}=0$, (${a_{k,2j}}=0$, ${a_{k,2j-1}}=0$) correspondingly. So given the generated graph with node locations, this sparse coefficient matrix can be calculated efficiently. 

Note that the location perturbations of the anchors are always zero, i.e., $\{\Delta \mathbf{p}_i = 0$ for $ i=\{1, \cdots, m\}$. Therefore, we can partition the coefficient matrix into two parts: $\mathbf{B}$ and $\mathbf{A}$. The first part $\mathbf{B}$ contains contains $2m$ columns and the second part $\mathbf{A}$ contains $2(n-m)$ columns. Then:
\begin{equation}
 \mathbf{\Delta D} = \left[ {\mathbf{B},\mathbf{A}} \right]\left[ \begin{array}{l}
\mathbf{0}\\
\mathbf{\Delta P}
\end{array} \right] = \mathbf{ A} \cdot \mathbf{\Delta P}
\label{A}
\end{equation}
where $ \mathbf{\Delta D} $ is the vector of edge length perturbation and  $ \mathbf{\Delta P} $ is the location perturbations of the ordinary nodes. 
\begin{definition}[Ranging Sensitivity Matirx: RSM]
The coefficient matrix $\mathbf{A}$ derived from (\ref{A}), which has $|E|$ rows and $2(n-m)$ colunms is defined as the  ranging sensitivity matrix (RSM), which represents the linear mapping from the node location perturbation to the edge length perturbation. 
\end{definition}
\subsection{Sensitivity Indicated by Condition Number of RSM} 
Recall the linear algebra, for a linear equation $Ax=b$,  the condition number  of $A$, denoted by $\kappa (A)$   indicates the maximum ratio of the relative error in $x$ divided by the relative error in $b$ \cite{cline_estimate_1979}. It measures how much the output value of the equation $x$ can be changed for a small change in the input $b$.  So in our problem,  we can evaluate the sensitivity of $ \mathbf{\Delta P} $ over $ \mathbf{\Delta D} $ based on the condition number of  $\mathbf{A}$. 
We first look at a property of $\mathbf{A}$:
\begin{theorem}
If the under investigating graph is  infinitesimal locally rigid, and there are $m\ge 2$ non-collocated anchors, then $\mathbf{A}$ has full rank $2(n-m)$. 
\end{theorem}
\begin{proof}
Suppose $\mathbf{A}$ has rank less than $2(n-m)$,  given a $ \mathbf{\Delta D} $, there will be infinite number of solutions of $ \mathbf{\Delta P} $. This contradicts with the condition that the under investigation graph is infinitesimal locally rigid and have been sticked by the  $m\ge 2$ anchors at the 2D plane (without the freedom of continuous transition and  rotation).  So $\mathbf{A}$ has full rank $2(n-m)$ when $m\ge 2$. 
\end{proof}

In graph generation in $2D$ space, we always need to name at least two anchors to define the local coordinate system. So $m\ge 2$ in practice. Then by SVD decompostion of $\mathbf{A}$, we can  get  $2(n-m)$ number of  non-zero eigenvalues of $\mathbf{A}$, denoted by ${\sigma_1(\mathbf{A}) \ge \sigma_2(\mathbf{A}) \ge \cdots \ge \sigma_{2(n-m)}(\mathbf{A}) }$, which are sorted in descending order. Then the condition number of $\mathbf{A}$ can be calculated as following:
\begin{definition} [Condition Number of $\mathbf{A}$]
$\kappa (\mathbf{A}) = \frac{{{\sigma _1}(\mathbf{A})}}{{{\sigma _{2(n-m)}(\mathbf{A})}}}$ is defined as the condition number of $\mathbf{A}$, which measures  the sensitivity of $ \mathbf{\Delta P} $ over $ \mathbf{\Delta D} $.
\end{definition}

$\kappa (\mathbf{A})$ is an indicator of the location robustness of nodes to the ranging noises.  When $\kappa (\mathbf{A})$ is large, it means  $\mathbf{A} $ is ill conditioned. Small changes in $ \mathbf{\Delta D} $ may cause large derivation of $ \mathbf{\Delta P} $ . When $\kappa (\mathbf{A})$ is small, the equation is well conditioned. Small changes of $ \mathbf{\Delta D} $ will not much deviate  $ \mathbf{\Delta P} $, which means the realized locations of nodes are robust to the ranging noises. 

\begin{figure}[t]
\begin{center}
\includegraphics[height=0.5in]{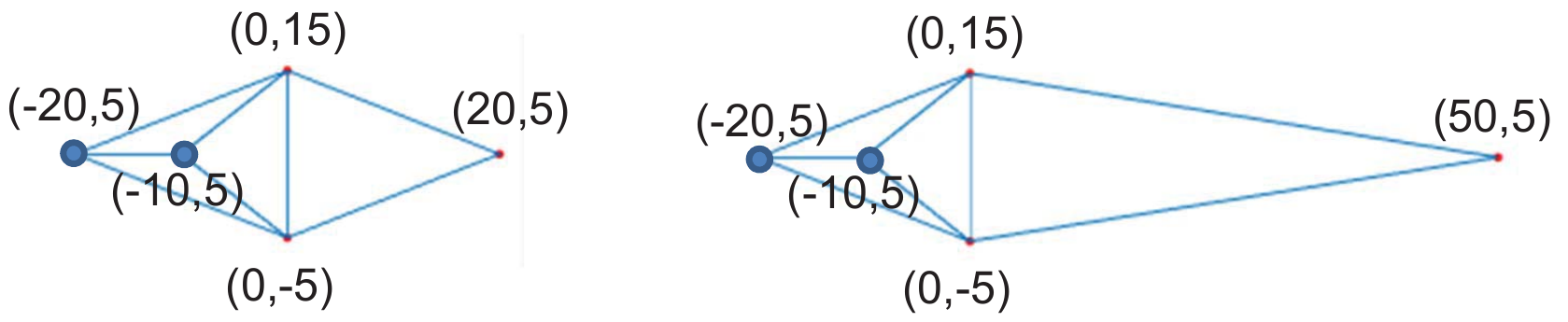}
\caption{ (a) an example topology.  (b) a topology which is more sensitive to ranging error. }
\label{example}
\end{center}
\end{figure}

\subsection{Local Deforming Sensitivity Evaluation Algorithm}
Given a generated graph, the sensitivity of local deforming to ranging noises can therefore be efficiently evaluated. 
Note that the most complex step is to calculate the eigenvalues, which has complexity $O(|E|\left(2(n-m)\right)^2)$. So the worst-case complexity is  $O(n^4)$, because there can be at most $n(n-1)/2$ edges.

 Fig.\ref{example} gives two example graphs to illustrate their condition numbers  for better understanding the graph's sensitivity to the ranging noises. In the graphs, the nodes at $(-20,5), (-10,5)$ are used as anchors. By evaluating $\kappa(\mathbf{A})$, the condition number of the graph in Fig.\ref{example}(a) is 13.68, while the condition number of the graph in Fig.\ref{example}(b) is 24.81.  This coincides with our general knowledge that the right graph is more sensitive to ranging noises. 


\section{Robust Component Localization }
Based on the above analysis of \emph{risk of flip ambiguity} and \emph{sensitivity of local deforming}, efficient scheme for robust component-based localization can be designed. Given the observed distance matrix $D$ among nodes, traditionally, component generation is to find rigid components to partition the graph;  component realization realizes each rigid component.  A rigid component is realized in a local coordinate system and may have finite number of realization candidates. Component mergence is to merge the candidate realizations of components based on the links among the components to finally realize the graph, if the graph is uniquely localizable\cite{aspnes_theory_2006}. The flow of the traditional component-based localization is shown in Fig.\ref{component}(b).   An example of the component generation and realization is shown in Fig.\ref{component}(a). Two components are generated from the distance matrix. It can be seen that every node in one component does not have two links connected to the other component. The second component has two candidate realizations because it may have a flip ambiguity. 
\begin{figure}[t]
\begin{center}
\includegraphics[height=1in]{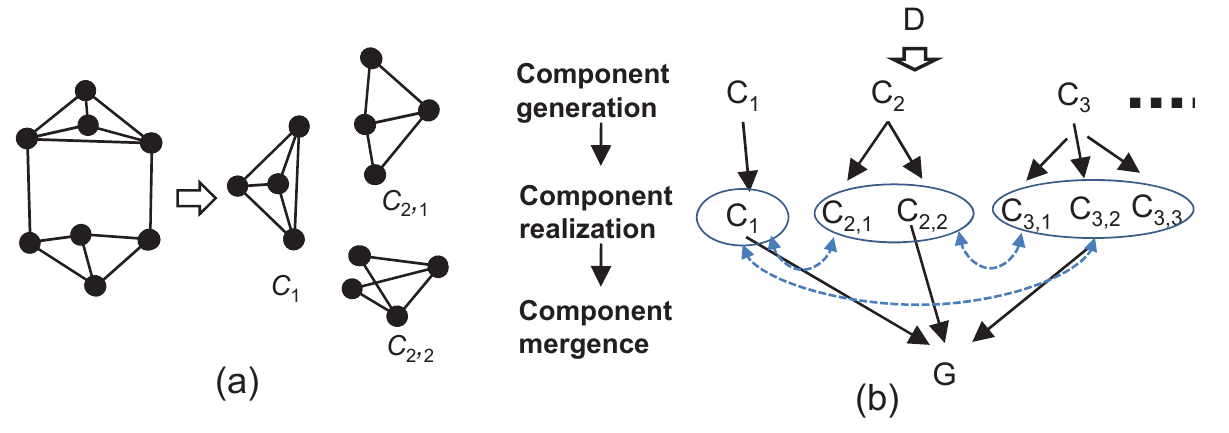}
\caption{ (a) An example of component generation and realization.  (b) Working flow of component-based localization. }
\label{component}
\end{center}
\end{figure}

The robust component mergence against noise was investigated in \cite{wang_etoc:_2010}. Therefore, this paper focuses on how to generate robust components, in which, the candidate locations of nodes are robust against noise. It will be designed based on the flip risk and RSM condition number evaluation. 
\subsection{Robust Component Generation and Realization (RCGR)}
The evaluations of flipping risk and local deforming sensitivity are integrated into the component generation and realization process. 
The overall routine is that: 
\subsubsection{Component Generation} The set of components $C$ is initialized as $\emptyset$. Then from distance matrix $D$, we firstly select three anchors with robust topology to initialize a component $C_1$, and then find nodes that have at least two edges connected to $C_1$ to expand $C_1$. The process repeats to add node to $C_1$ until no other nodes having two links with $C_1$. 
The heuristic method to  select the initial triangle is to find three vertex in $G\setminus C$, such that
\begin{equation}
\{v_i, v_j, v_k\} = \arg\max_{i,j,k} \{d_{i,j} \cdot d_{i,k}\cdot d_{j,k}\}
\label{anchor}
\end{equation}
because a large equilateral triangle is robust to both flip and local deforming errors. 
 \begin{figure*} [t]
   \begin{minipage}[t]{0.247\linewidth} 
     \centering 
     \includegraphics[height=1.5in]{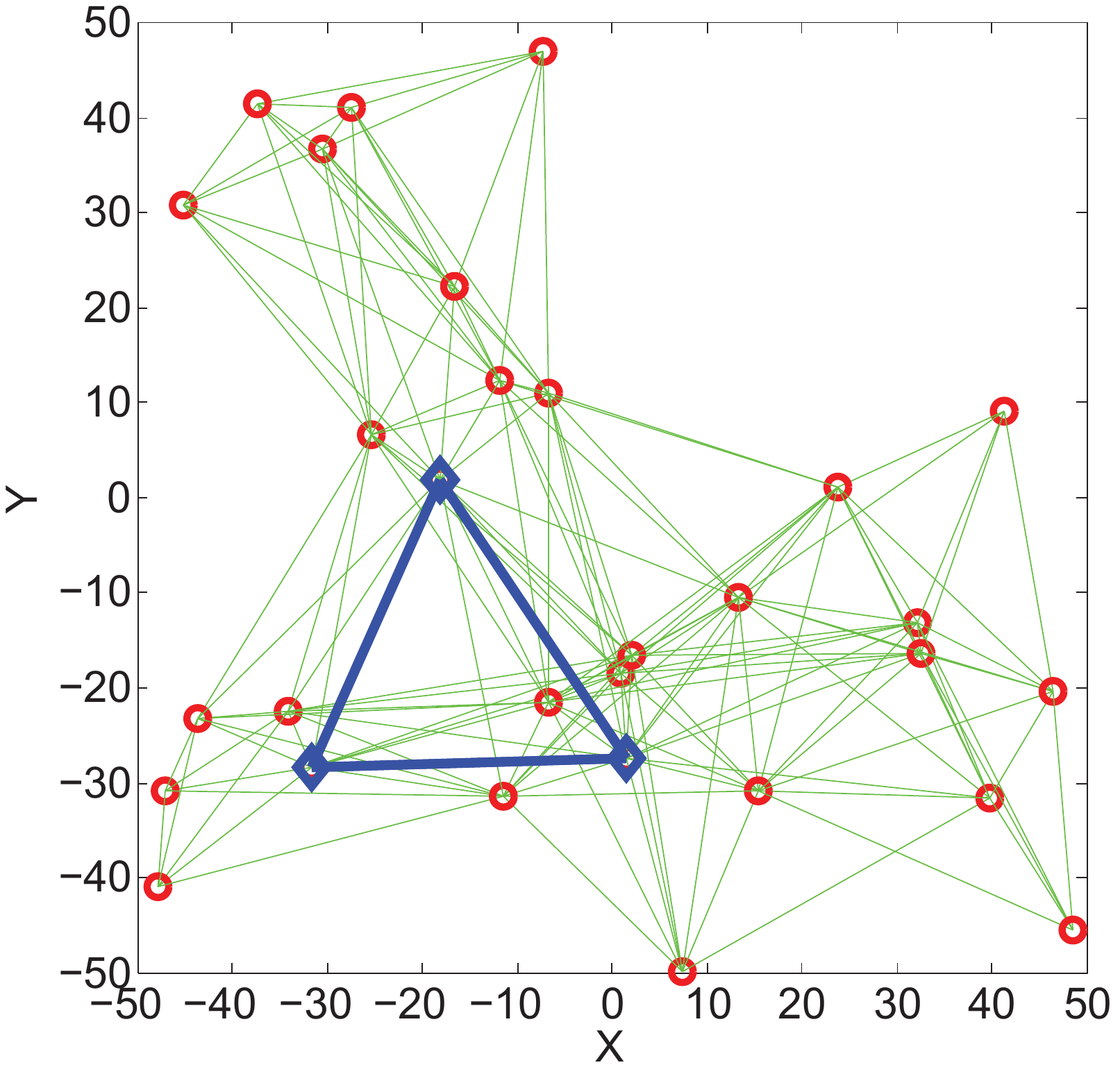} 
     \caption{Network topology  and anchor selected by the heuristic in (\ref{anchor})} 
     \label{eva:1} 
   \end{minipage}
   \begin{minipage}[t]{0.247\linewidth} 
     \centering 
     \includegraphics[height=1.5in]{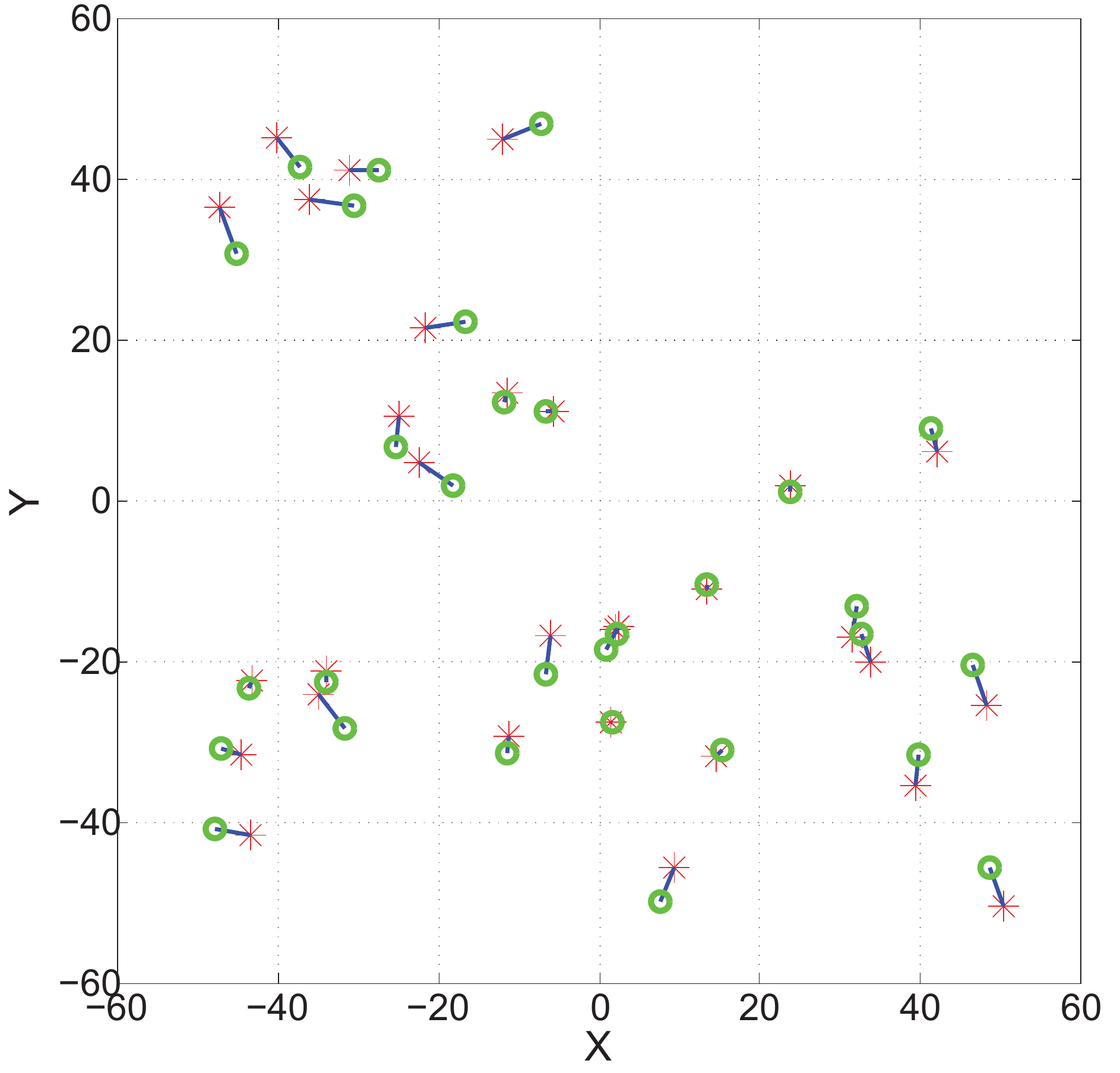} 
     \caption{Results by SDP } 
     \label{eva:2} 
   \end{minipage}%
   \begin{minipage}[t]{0.247\linewidth} 
     \centering 
     \includegraphics[height=1.5in]{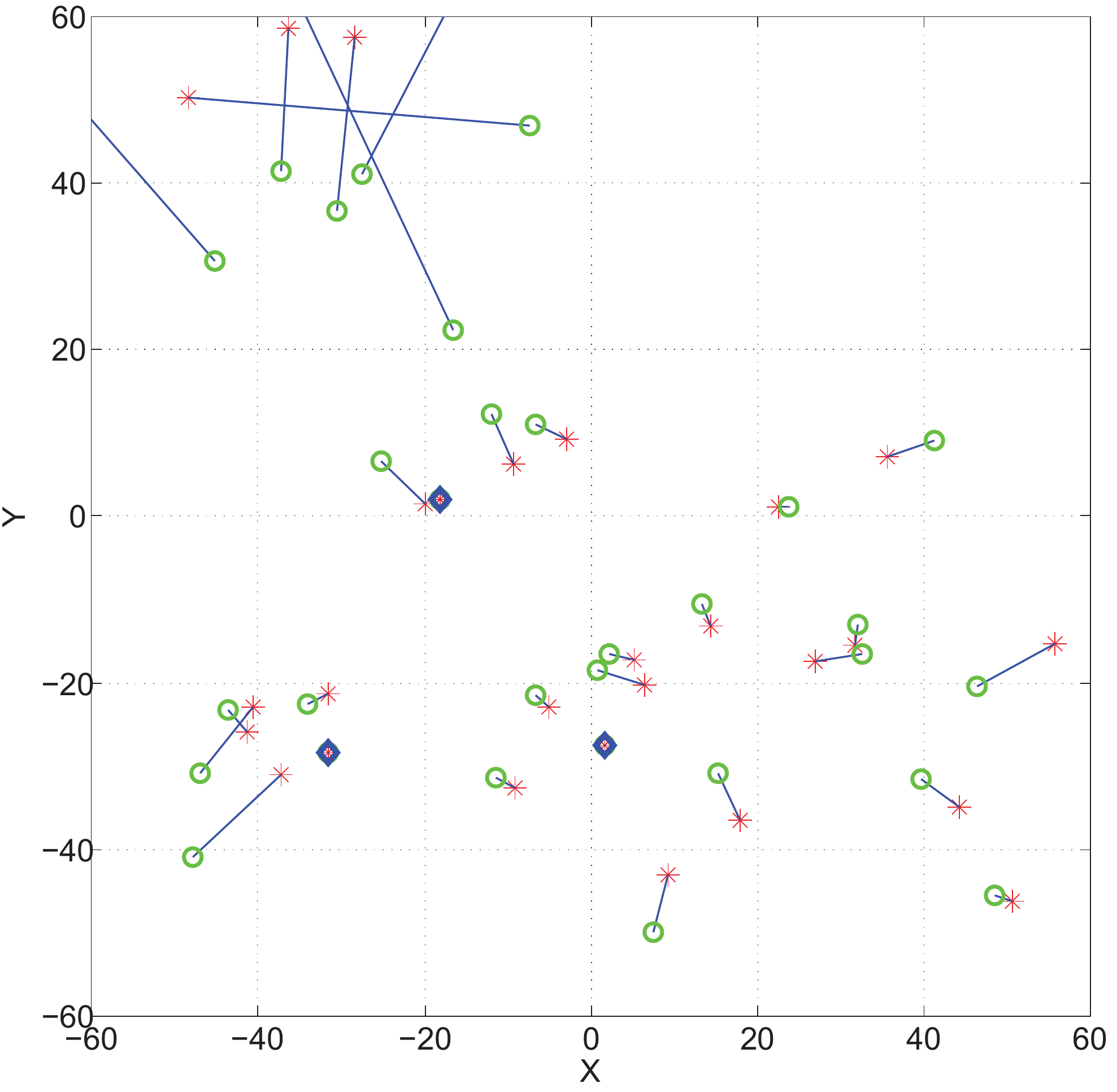} 
     \caption{Results by CALL,  one component} 
     \label{eva:3} 
   \end{minipage}
   \begin{minipage}[t]{0.247\linewidth} 
     \centering 
     \includegraphics[height=1.5in]{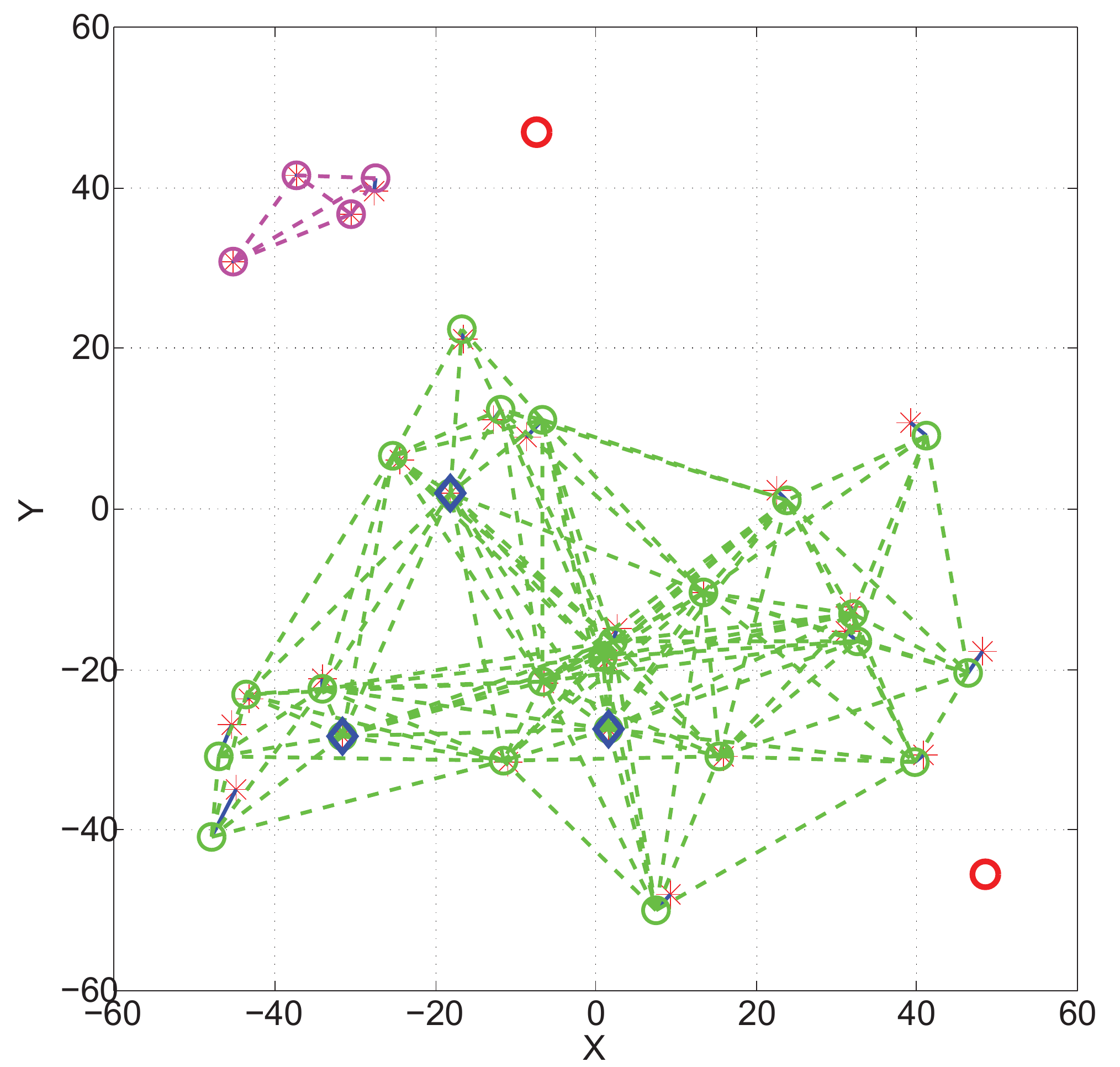} 
     \caption{Results by RCGR,  two components} 
     \label{eva:4} 
   \end{minipage} 
 \end{figure*}
\subsubsection{Robust Realization} Node locations in $C_1$ are realized sequentially. The realized node set $R_1$ starts from the anchors. The remaining node in $C_1$ with the smallest RSM condition number will be realized firstly, and be added into $R_1$.  The process repeats to realize the node with the smallest RSM condition number one per time and adds it to $R_1$. The process stops when the condition number of all the remaining nodes are larger than a threshold.  
\subsubsection{Component Refinement to Exclude Unreliable Nodes} The remaining unrealized nodes in $C_1$ will be excluded from $C_1$ to be returned to $G\setminus C$. $C_1$ equals only to $R_1$. Then MIRRORs are found in the realization of $R_1$ to list all possible location candidates of $\{L(C_1)\}$. 
\subsubsection{The Recursion}  In  $G\setminus C$, another  robust triangle will be selected as anchors, to generate $C_2$, and the process repeats as from 1) to 3). The overall process repeats until no triangles can be found in $G\!\setminus\!C$ .  
\subsection{The RCGR Algorithm}
The algorithm flow is given in Algorithm\ref{algorithm4}. Line 3 to 6 are for component generation; line 7 to 14 robustly realize the nodes in the component.  In line 15 to 16, the nodes that cannot be robustly realized will be excluded from the component. In line 17, candidate realizations are generated by finding MIRRORs.  
  \begin{algorithm}[htbp]
  \caption{RCGR Algorithm}
  \begin{algorithmic}[1]
  \small
  \REQUIRE  The graph $G=(V,E)$, Distance matrix $D$.  
  \ENSURE   Components $C=\{C_1, \cdots, C_K\}$; Realization candidates for components $\{\{L(C_1)\}, \cdots, \{L(C_K)\}\}$
  \STATE  Initialize $C=\emptyset$; $k=1$; not grouped part $U_k = G\!\setminus\!C$; 
  \WHILE{$U_k$ contains triangle}
         \STATE  Find a robust triangle as anchors to initialize $C_k$;  
       \WHILE  {(existing a node $u$ in $U_k\!\setminus\!C_k$ having two edges connected to nodes in $C_k$)}
       \STATE  $C_k = C_k \cup u$; 
        \ENDWHILE
       \STATE  In $C_k$, add anchors to the realized node set $R_k$;  
       \WHILE {($R_k$ having unrealized neighbors in  $C_k$)}
              \FOR{ $u$ be a neighbor of $R_k$ (denoted by $N(R_k)$)}
       		       \STATE  Generate location of $u$ by $R_k$ and evaluate $\kappa(u)$.
       		 \ENDFOR
             \IF{$\kappa(u_o)$ is the smallest in $N(R)$ and $\le T_\kappa$}
                  \STATE    $R_k=R_k\cup u_o$;  Save location $L(R_k)$. 
             \ENDIF     
                 \IF {$\min\{\kappa(u), u\in N(R)\} > T_\kappa$ }
                       \STATE  Break the while loop of realization and exclude unrealizable nodes from $C_k$.  $C_k\!=\!R_k$, $U_k\!=\!U_k\!\setminus\!R_k$.  
                       \STATE  $C= C\cup C_k$. Find MIRRORs in $L(R_k)$ to generate all realization candidates for $C_k$. 
                \ENDIF
        \ENDWHILE
      \STATE  k=k+1;
\ENDWHILE
  \end{algorithmic}
  \label{algorithm4}
  \end{algorithm}
\section{Simulation Results}
Simulations were conducted to evaluate the effectiveness of the proposed metrics and algorithms. Since the number of flipping ambiguity and the sensitivity of local deforming are mainly used in the RCGR algorithm, we focus on their impacts on the component generation and the contributions to the location accuracy improvement.

\subsection{Experiment Settings}
We evaluate our proposed algorithm by comparing with two state-of-the-art network localization algorithms. 1) Semidefinite programming (SDP)-based localization algorithm in \cite{biswas_semidefinite_2006},  which used SDP-based approach with regularization to solve the network localizaiton problem. Since it uses centralized optimization, its location accuracy can be thought as the upper bound that the component-based method can achieve. 2) Component-bAsed Localization aLgorithm (CALL) \cite{wang_component_2008}, which is the state-of-the-art in component-based localization. 

In simulation, we generate $n$ nodes uniformly random in an area of $100(m)*100(m)$. The ranging radius $r$ of each node varies to keep the distance matrix sparse. Based on the distance matrix, components generation and realization are carried out by RCGR and CALL  respectively and their localizaiton results are compared with that of SDP-based localization. 
\subsection{Robust Component Generation}
For a network of 30 nodes, with ranging radius $30m$, each edge $d_{i,j}$ is affected by multiplicative ranging noise,  i.e., $n_{i,j} = d_{i,j}(1+L_1/10)$, where $L_1 \sim N(0,1)$.  The network topology is shown in Fig.\ref{eva:1}.  The localization results given by SDP-based method are shown in Fig.\ref{eva:2}, which shows the centralized optimization gives rather accurate location results.  The localization results by CALL algorithm are shown in Fig.\ref{eva:3}. The nodes are formed into one component and the localization results of the nodes comparing to their true locations are shown. The true locations are represented by circles and the estimated locations are  by stars. We can see that the nodes far from the anchors, which are localized later in the sequential localizing process suffer large accumulated errors.  

The location results and the formed components by RCGR are shown in Fig.\ref{eva:4}.  RCGR adaptively organizes the graph into two components by online evaluating the component's local deforming sensitivity. The threshold for $\kappa(A)$ was set to 4 in experiments.  The locations of the nodes in the two components are accurately calculated via mergence by the several links between the two rigid components\cite{wang_etoc:_2010}. It can be seen the accuracy is nearly comparable to that of the centralized approach. By extensive simulations, we find RCGR can always adaptively organize components based on the evaluated robustness of the components. It  forms smaller but robust components while keeping relative good localizing accuracy in each formed component.
 \begin{figure} [t]
   \begin{minipage}[t]{0.48\linewidth} 
     \centering 
     \includegraphics[height=1.45in]{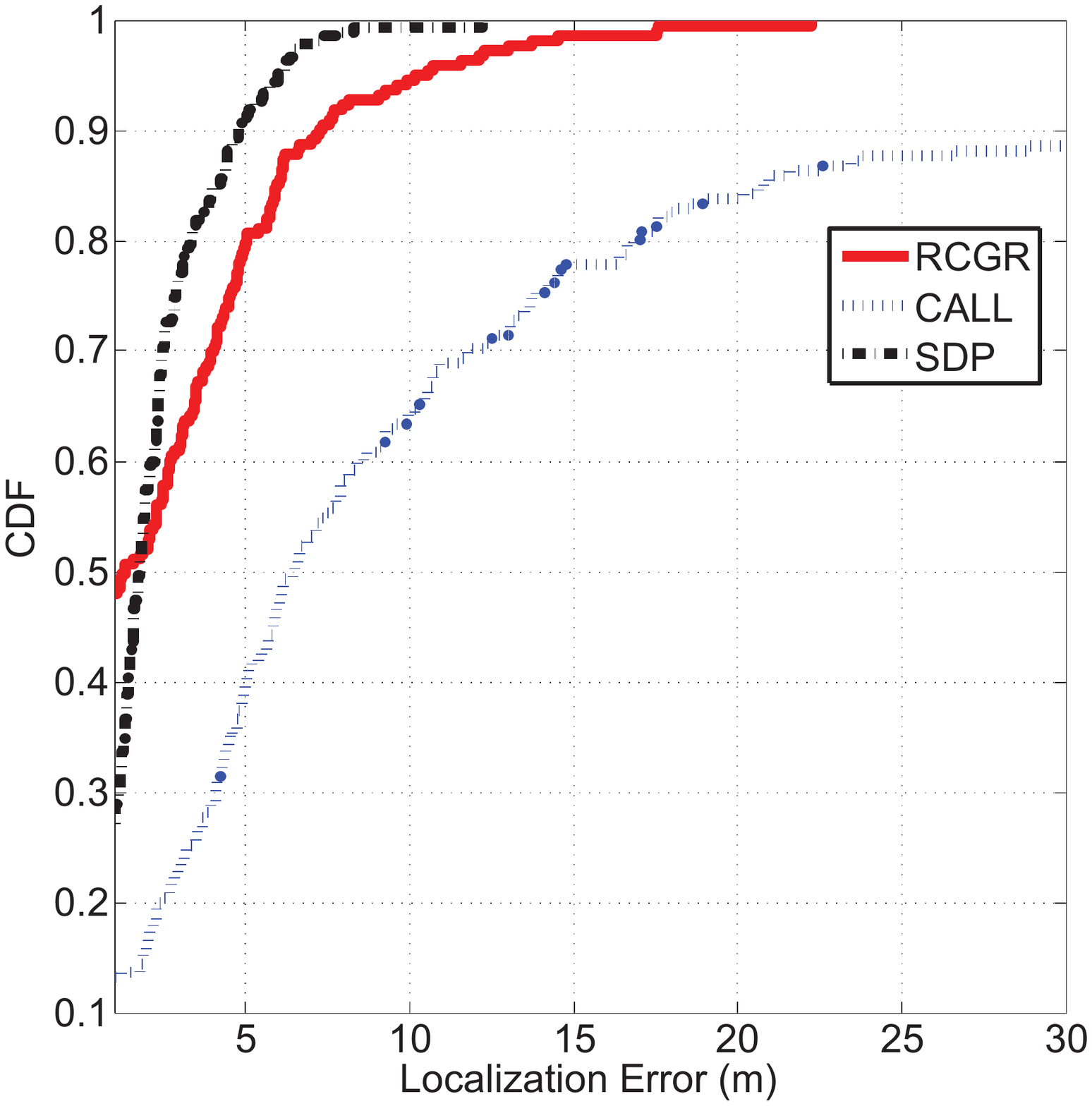} 
     \caption{The CDF of location errors} 
     \label{eva:5} 
   \end{minipage}
   \begin{minipage}[t]{0.48\linewidth} 
     \centering 
     \includegraphics[height=1.45in]{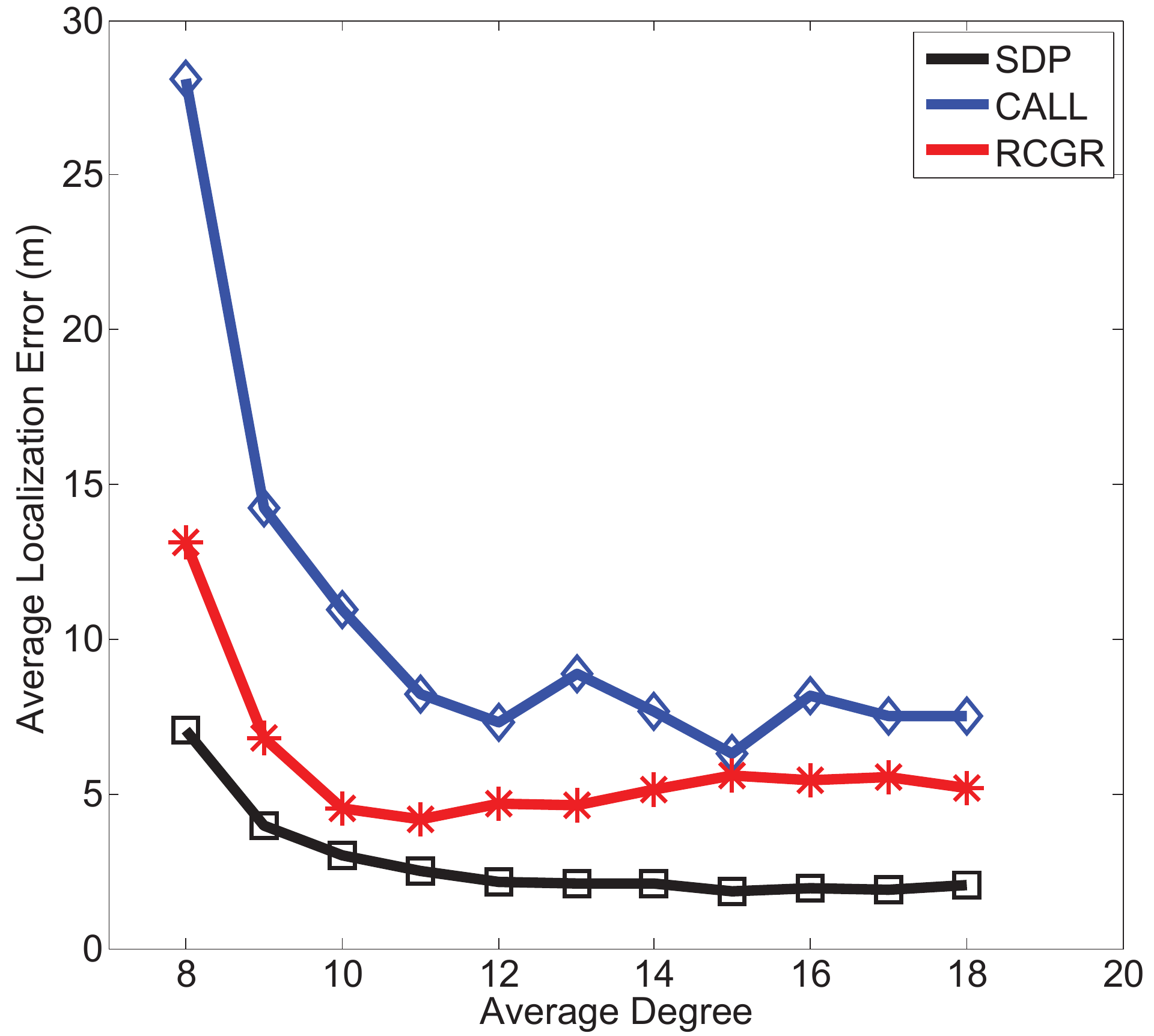} 
     \caption{Accuracy vs. average node degrees } 
     \label{eva:6} 
   \end{minipage}%
 \end{figure}
\subsection{Location Accuracy Improvement}
By varying node number from 30 to 50, we generate random networks for 30 times,  and run the  localization algorithms in each network. In each experiment, the location error of each node is evaluated by the distance from its true location to the estimated location. The cumulative error distribution of the location errors of nodes in the thirty experiments are summarized in Fig.\ref{eva:5}. We can see the location accuracy of RCGR is much better than CALL and only a little worse than SDP. This shows the effectiveness of  RCGR.  

We also  evaluate the localization performances for networks regarding to different average node degrees. When the average node degree varies from 8 to 18, the average localization errors of the three algorithms are compared in Fig.\ref{eva:6}. It can be seen that, as node degrees increase, the location accuracies of all the three algorithms become better. This is because in each step each node can be trilaterized by many localized nodes. RCGR shows better performance than CALL,  and SDP has the best accuracy.  But note that, as the average degree increases, the accuracy performance of RCGR  drops a little bit. This is because larger components are formed due to  denser edge connections, which increases the error accumulation. Experiments in larger scale networks show the similar results, which will be not detailed for space limitation.

\section{Conclusion}
This paper has presented two metrics to evaluate the robustness of a rigid graph under ranging noises: 1) a MIRROR detection-based method to evaluate the possible number of flip ambiguities in a rigid graph under ranging noises; 2) a method to calculate the condition number of Ranging Sensitivity Matrix to evaluate the graph's sensitivity of local deforming regarding to ranging noises. Based on these two metrics, a robust component generation and realization (RCGR) algorithm was proposed, which can form robust components adaptively based on the online component robustness evaluation and can improve the accuracy of component-based localization.  These results can be further extended into network localization in 3D space . The robustness metrics may also be exploited to other problems that an be modeled by graph realization under noise, such as robot group formation. It can also be investigated in Simultaneously Locating and Mapping (SLAM) for dealing with error accumulation while keeping low computing cost.  

\bibliographystyle{abbrv}
\bibliography{rigid}

\end{document}